\newcommand{\version}{May 31st, 2018}
\documentclass[11pt,a4paper,reqno]{amsart}
\usepackage{amsthm,amsmath,amsfonts,amssymb,amsxtra,appendix,bookmark,dsfont,latexsym}

\makeatletter
\usepackage{hyperref}
\usepackage{delarray,color}
\usepackage{euscript}

\newcommand{\bdm}{\begin{displaymath}}
\newcommand{\edm}{\end{displaymath}}
\newcommand{\bdn}{\begin{eqnarray}}
\newcommand{\edn}{\end{eqnarray}}
\newcommand{\bay}{\begin{array}{c}}
\newcommand{\eay}{\end{array}}
\newcommand{\ben}{\begin{enumerate}}
\newcommand{\een}{\end{enumerate}}

\newcommand{\R}{\mathbb{R}}

\newcommand{\eps}{\varepsilon}

\newcommand{\intR}{\int_{\R ^2}}

\newcommand{\PsiLau}{\Psi_{\rm Lau}}
\newcommand{\cLau}{c _{\rm Lau}}

\newcommand{\cH}{\mathcal{H}}
\newcommand{\cW}{\mathcal{W}}

\newcommand{\ETF}{\mathcal{E} ^{\rm TF}}

\newcommand{\TFmin}{\sigma ^{\rm TF}}
\newcommand{\TFpot}{\Phi ^{\rm TF}}
\newcommand{\STF}{\Sigma ^{\rm TF}}

\newcommand{\Vnuc}{V_{\rm nuc}}

\newtheorem{teo}{Theorem}
\newtheorem{lem}[teo]{Lemma}
\newtheorem{cor}[teo]{Corollary}

\theoremstyle{definition}

\theoremstyle{remark}

\newcommand{\beq}{\begin{equation}}
\newcommand{\eeq}{\end{equation}}

\numberwithin{equation}{section}

\begin{document}

\title{Rigidity of the Laughlin liquid}

\author{Elliott H. Lieb}
\address{Departments of Mathematics and Physics, Princeton University,
Princeton, NJ 08544, USA}
\email{lieb@princeton.edu}

\author{Nicolas Rougerie}
\address{Universit\'e Grenoble Alpes \& CNRS, LPMMC (UMR 5493), B.P. 166, F-38042 Grenoble, France}
\email{nicolas.rougerie@lpmmc.cnrs.fr}


\author{Jakob Yngvason}
\address{Faculty of Physics, University of Vienna, Boltzmanngasse 5, A-1090 Vienna, Austria}
\email{jakob.yngvason@univie.ac.at}



\date{\version}

\begin{abstract}
We consider general $N$-particle wave functions that have the form of a product of the Laughlin state with filling factor $1/\ell$ and an analytic function of the $N$ variables. This is the most general form of a wave function that can arise through a perturbation of the Laughlin state by  external potentials or impurities, while staying in the lowest Landau level and maintaining the strong correlations of the original state. We show that the perturbation can only shift or lower the 1-particle density but nowhere increase it above a maximum value. Consequences of this bound for the response of the Laughlin state to external fields are discussed. 
\end{abstract}

\maketitle
 \begin{center} {\it Dedicated to J\"urg Fr\"ohlich, Thomas Spencer, and Herbert Spohn}
 \end{center}
\bigskip

\section{The Laughlin phase}

In theoretical studies of the fractional quantum Hall effect (FQHE)~\cite{Girvin-04,Goerbig-09,Frohlich-92,Frohlich-95,Jain-07,StoTsuGos-99} Laughlin's wave function(s)~\cite{Laughlin-83,Laughlin-87,Laughlin-99} play a fundamental role. There is such a function for every positive integer $\ell$ and it can be written, in units where the magnetic length is $1/\sqrt{2}$,  as
\begin{equation}\label{laufunc}
\PsiLau =\cLau \prod_{i<j}(z_i-z_j)^{\ell} e^{-\sum_{i=1}^N|z_i|^2/2}
\end{equation}
where the $z_i\in\mathbb C$ are the positions of $N$ particles moving in $\mathbb R^2$, identified with the complex plane, and the constant $\cLau$ is a normalization factor (depending on $N$ and $\ell$). For fermions, $\ell$ is odd and $\geq 3$ (the case $\ell=1$ corresponds to noninteracting fermions), while for bosons $\ell\geq 2$ is even. Bosonic wave functions of this type are potentially relevant for atomic gases in artificial magnetic fields~\cite{Cooper-08,BloDalZwe-08,LewSei-09,MorFed-07,RonRizDal-11,Viefers-08}. The analysis below  applies to Laughlin states of both
symmetry types.

The Laughlin state~\eqref{laufunc} is a special kind of wave function in the Lowest Landau level (LLL) of a Hamiltonian with a strong magnetic field perpendicular to the plane where the particles move.  The general form of such functions is
\begin{equation}\label{genfunc}
\Psi(z_1,\dots,z_N)= A(z_1,\dots, z_N)\, e^{-\sum_{i=1}^N|z_i|^2/2}
\end{equation}
with $A$ analytic, antisymmetric for fermions and symmetric for bosons. Also square integrability is required, but $A$ is not restricted to polynomials. 
If there is strong repulsive two-body interaction between the particles, the Laughlin function ~\eqref{laufunc} is a natural trial function for low energy states. The factors $(z_i-z_j)^{\ell}$ indeed suppress interactions by producing strong correlations between the particle positions. For certain zero-range interactions,~\eqref{laufunc} is even an exact ground state~\cite{Haldane-83,TruKiv-85,PapBer-01}. 

The success of Laughlin's theory of the FQHE fractions $1/\ell$, based on~\eqref{laufunc}, depends crucially on the fact that the Laughlin wave function behaves as an {\it incompressible liquid}, whose response to perturbations and external fields is extremely rigid. This in fact has two distinct aspects:

\smallskip

\noindent \textbf{1.}~The Laughlin wave function is an approximate ground state for the many-body Hamiltonian, and its energy is separated from the rest of the spectrum by a gap.

\smallskip 

\noindent \textbf{2.}~Modifications of the Laughlin wave function that stay within the ground eigenspace of the many-body Hamiltonian 
cannot increase the local one-particle density beyond a fixed value.

\smallskip  

The main evidence for these facts is experimental and numerical. 
For model zero-range interactions, the justification of Property \textbf{1} would amount to a proof of an $N$-independent spectral gap, see e.g.~\cite[Section~2.2]{LewSei-09} or~\cite[Section~2.1]{RouSerYng-13b} and references therein.  We are not aware of a solution to this important problem, but some exact results for low values of $N$ can be found in \cite{MasMatOuv-07,MasMatOuv-13}.

Property \textbf{2} is not always recognized as a separate issue, but is also crucial for Laughlin's original argument supporting quantization of the Hall conductivity~\cite{Laughlin-81,Laughlin-99} and for FQHE physics in general. 
Its importance lies in understanding the effect of an external potential that would try to concentrate the density as much as possible in energetically favorable places. 

If the energy gap in Property \textbf{1} is sufficiently large we can exclude a jump across it and restrict attention to wave functions with the same interaction energy as the Laughlin function. We are thus led to study the set of all normalized wave functions of the form
\begin{equation}\label{fullcorr} 
\Psi_F (z_1,\dots, z_N)= F(z_1,\dots, z_N)\PsiLau (z_1,\dots, z_N)
\end{equation}
with $F$ analytic and symmetric under exchange of the $z_i$.   This form exhausts the class of functions that minimize the magnetic kinetic energy and at the same time avoid repulsive interactions by vanishing at least as $(z_i-z_j)^\ell$ as $z_i$ and $z_j$ come together. In the bosonic case and with $\ell= 2$ these are exactly the ground states of the contact interaction~\cite[Section~2.1]{RouSerYng-13b}. We shall refer to the class of states of the form \eqref{fullcorr} as {\em fully correlated states}.

Physically,~\eqref{fullcorr} includes the addition of \lq quasi-holes\rq\, (zeros of the wave-function) to the Laughlin state, essentially arbitrary correlations between the particles' and quasi-holes' locations being allowed. It is intuitive that this leads to a decrease of the \emph{global} density. It is far from obvious, however,  that no \emph{local} increase of the density can occur anywhere. Several quasi-holes arranged tightly on a circle could, perhaps, increase the density inside the circle. Moreover, $F$ need not contain any zeros at all like, e.g., $\exp(c\sum_i z_i^2)$, which stretches the support of the density in one direction and shrinks it in another.

%

In this paper, we report on recent density bounds for fully correlated states that demonstrate the validity of Property \textbf{2} without invoking Property \textbf{1}. They hold essentially on length scales $O(N^{1/4})$, much smaller than the full extent of the liquid, which is of order $N^{1/2}$, and are related to earlier partial results proved in~\cite{RouYng-14,RouYng-15}. The full proofs are somewhat involved and presented elsewhere~\cite{LieRouYng-17}, but we sketch the main arguments below and discuss physical applications.  

\section{Incompressibility estimates}

In his pioneering paper~\cite{Laughlin-83}, Laughlin already argued that the one-particle density of the state \eqref{laufunc} has the form of a circular droplet of radius $\sqrt {\ell N}$ where the density takes the constant value $(\pi\ell)^{-1}$. The argument was based on the plasma analogy, where the absolute square of the wave function is written as the Gibbs distribution of a classical 2D Coulomb gas, and subsequently treated by a mean-field approximation. 

Recent mathematical analysis~\cite{RouSerYng-13a,RouSerYng-13b, {RouYng-17}} has confirmed the validity of this approximation for the Laughlin states and Laughlin's   quasi-hole states where the prefactor $F$ is a product in which each factor depends on a {\em single} variable, i.e., wave functions of the special form
\begin{equation}\label{eq:QH state}
\Psi_{f^{\otimes N}} (z_1,\dots, z_N)=  \prod_{j=1} f(z_j) \PsiLau (z_1,\dots, z_N) 
\end{equation}
with $f$ being a polynomial. These results however do not imply density bounds for the {\em general} state~\eqref{fullcorr}. The latter could, in particular, be a linear superposition of functions of the form~\eqref{eq:QH state}, and the density of a linear superposition needs not coincide with the linear superposition of the densities, because of quantum-mechanical interferences.

The analysis of~\cite{RouSerYng-13a,RouSerYng-13b}  was generalized to other prefactors of a special kind in~\cite{RouYng-14}. A common feature that emerged was an upper bound on the one-particle density of magnitude $(\pi\ell)^{-1}$, which is the density of the Laughlin state itself. Such a bound was called an \emph{incompressibility estimate} in \cite{RouYng-14} because it is a manifestation of the resistance of the Laughlin state against attempts to compress its density. It relies essentially on the strong correlations of the Laughlin function and the analyticity of $F$ in~\eqref{fullcorr} which, physically,  is due to the strong magnetic field confining the particles to the LLL. Without such a field and with weak correlations, the electron density in a crystal can be arbitrarily high locally, due to constructive interference of Bloch waves.

The mean-field methods of \cite{RouSerYng-13a,RouSerYng-13b,RouYng-14} are not applicable to general prefactors $F$.
The question of a bound on the density for the general case was treated in~\cite{RouYng-15}  with an entirely different technique, rooted in 2D potential theory. The bound obtained was four times the expected optimal value $(\pi\ell)^{-1}$ however. In this paper we explain that an improved version of the potential theoretic method leads to the correct optimal bound for \emph{arbitrary}~$F$. 
%

As indicated by numerical studies~\cite{Ciftja-06}, the incompressibility bound for the density cannot be expected to hold pointwise for finite $N$.  We prove, however, that it holds for local averages. Denote 
\begin{equation}\label{eq:original density}
\rho_F (z) = N \int_{\R ^{2(N-1)}}  \left| \Psi_F \left(z, z_2,\ldots,z_N \right) \right| ^2 dz_2 \ldots dz_N
\end{equation}
the particle density of the state~\eqref{fullcorr}. Our main result is

\begin{teo}[\textbf{Density bound}]\label{main thm0}\mbox{}\\ 
For any $\alpha > 1/4$ and any disk  $D$ of radius $N^{\alpha}$ we have, uniformly in the choice of $F$,  
\begin{equation}\label{bound}
\int_{D} \rho_F \leq \frac{1}{\pi\ell} |D| (1+o(1))
\end{equation}
where $|D|$ is the area of the disk and $o(1)$ tends to zero as $N\to \infty$.
\end{teo}

Recall that the magnetic length is $1/\sqrt{2}$. The average can thus be taken on ``mesoscopic scales'', $O(N^{1/4+\eps})$, much smaller than the full extent of the state, $O(N ^{1/2})$, but not quite down to the expected finest scale, $O(1)$. For realistic numbers~\cite{StoTsuGos-99} one can have Quantum Hall systems with $N$ of the order of $10^9$, in which case the ratio between the length scales $N^{1/4}$ and $N^{1/2}$ is about $10^{-2}$. A simple covering argument (reminiscent of the \lq cheese theorem\rq, see~\cite[Section~14.4]{LieSei-09} or~\cite[Theorem~14]{Lieb-76}) shows that~\eqref{bound} implies the analogous result for any open set, not just a disk. 

It follows that, for a continuous confining potential $V$, the smallest potential energy obtainable with a fully correlated state, i.e.,
\begin{equation}\label{eq:energy Nn}
E_V(\ell,N) := \min \left\{ \int_{\R ^2} V \rho_F, \; \Psi_F \mbox{ of the form~\eqref{fullcorr}} \right\},
\end{equation} 
is bounded below by the ``bathtub energy''~\cite[Theorem~1.14]{LieLos-01} 
\begin{equation}\label{bathtub}
E^{\rm bt}_V (\ell) := \min\left\{ \intR V \rho \: \Big| \: 0 \leq \rho \leq \frac{1}{\pi \ell},\ \intR\rho = N \right\},
\end{equation}
i.e. the minimum of the potential energy over all densities $\rho$ bounded by $(\pi\ell)^{-1}$. This lower bound 
\begin{equation}
E_V(\ell,N) \gtrapprox E^{\rm bt}_V (\ell) 
\end{equation}
holds for large $N$ provided $V$ varies only on scales much larger than $N ^{1/4}$.

The bound \eqref{bound} means that any compression of the particle density above the ``magic value'' $(\pi\ell) ^{-1}$ that one could imagine to accommodate the variations of an external potential would make us leave the class of fully correlated states, with corresponding increase in either the magnetic kinetic energy or the interaction energy. Assuming the spectral gap mentioned in Property~\textbf{1}, no such density bump is allowed. 
This justifies two things \emph{a posteriori}:

\smallskip

\noindent $\bullet$ That it is legitimate to neglect disorder in the sample and/or small external electric fields, as is done as a first approximation in the derivation of FQHE wave-functions. 

\smallskip 

\noindent$\bullet$ Laughlin's argument~\cite{Laughlin-81,Laughlin-87} (see also~\cite[Sections 4.4, 9.3 and 9.5]{Jain-07}) that switching on an electric current moves electrons transversally without creating any charge accumulation, and generates a Hall conductivity of value~$1/\ell$.

\smallskip

It has been proposed (see~\cite{Cooper-08,BloDalZwe-08,Viefers-08} for reviews) that Laughlin wave functions could be created in cold atomic gases, either by rapid rotation or by applying artificial magnetic fields. In this context, a magneto-optical trap  confines the gas. Some recent proposals to reach the Laughlin state~\cite{MorFed-07,RonRizDal-11} involve some non-trivial engineering of the trapping potential. How the Laughlin state responds is, therefore, of importance for the experimental set-up. 

In addition, the precursor of FQHE states in a rapidly rotating Bose gas is a Bose-Einstein condensate (see~\cite{LewSei-09,LieSeiYng-09} and references therein). Observing the distinctly flat profile of the Laughlin state would already be a strong indication of the transition to the FQHE regime. A more complete probe could be the response of the gas to variations of the trapping potential: the Bose condensate follows the trap by taking a Thomas-Fermi-like shape (see~\cite{AftBlaDal-05,AftBlaNie-06a,BlaRou-08} and references therein). The Laughlin state essentially does not respond to such variations, as exemplified by our main theorem. 

We also point out that a combination of Theorem~\ref{main thm0} with estimates obtained in~\cite{RouSerYng-13a,RouSerYng-13b} leads to the following improvement of~\cite[Corollary~2.3]{RouYng-14}:

\begin{cor}[\textbf{Optimization of the energy in radial traps}]\label{cor:radial}\mbox{}\\
Let $V(x) = |x| ^s$ with $s>0$. Then the potential energy within the class~\eqref{fullcorr} is minimized by the Laughlin state ($F=1$):
\begin{equation}\label{eq:incomp opt}
\lim_{N\to\infty} \frac{\intR V \rho_1  }{E_V(\ell,N)}  = \lim_{N\to\infty} \frac{E^{\rm bt}_V (\ell)}{E_V(\ell,N)} = 1
\end{equation}
where $E_V(N,\ell)$ and $E ^{\rm bt} _V (\ell)$ are defined in~\eqref{eq:energy Nn} and~\eqref{bathtub} respectively.
%
\end{cor}

It is remarkable that the Laughlin state stays an approximate minimizer in any power-law trap (the result actually holds for more general radial increasing potentials). No matter how steep and narrow a potential well one imposes, it is impossible to compress the Laughlin state while keeping the form~\eqref{fullcorr}, i.e., without jumping across the spectral gap. Extensions of Corollary~\ref{cor:radial} to general traps were recently proved in~\cite{RouYng-17}: For a trap of {\it arbitrary} shape, the minimal energy is asymptotically equal to the bathtub energy and can always be achieved by adding uncorrelated quasi-holes on top of Laughlin's function, i.e., by using wave functions of the form \eqref{eq:QH state}.

\section{Proof strategy: the exclusion rule}

We now turn to sketching the proof of Theorem~\ref{main thm0}. Details are given in the longer paper~\cite{LieRouYng-17}. It is convenient to change variables and consider the scaled $N$-particle probability density
\begin{equation} \label{eq:scaledvar}
\mu_F (Z_N) := N^{N} \left| \Psi_F \left(\sqrt{N} \: Z_N \right) \right| ^2
\end{equation}
corresponding to the wave-function~\eqref{fullcorr}. This has an extension $O(1)$ for the Laughlin state, $F= 1$. Here $Z_N$ stands for $(z_1,\dots, z_N)$. The scaled 1-particle probability density is
\begin{equation} 
\mu_F^{(1)}(z) = \int_{\R ^{2(N-1)}} \mu_{F} (z,z_2,\ldots,z_N)dz_2 \ldots dz_N 
= \rho_F \left( \sqrt{N}  z\right)\label{eq:muF 1}
\end{equation}

The first step is to write the $N$-particle density as a Gibbs factor (Laughlin's plasma analogy),
\begin{equation}\label{eq:plasma analogy}
\mu_F (Z_N) =\mathcal Z_{N}^{-1} \exp\left(\hbox{$ -\frac{1}{T}  H_N(Z_N)$} \right),
\end{equation}
with effective ``temperature'' $T=N^{-1}$ and the \lq Hamiltonian' 
\begin{equation}\label{ham}
H_N(Z_N)=\sum_{j=1} ^N |z_j| ^2  -\frac {2\ell}N \sum_{i<j} \log {|z_i - z_j|}+W_N (Z_N),
\end{equation}
where
\begin{equation} 
W_N (Z_N)=- \frac{2}{N} \log \left| F \left( \sqrt{N} \: Z_N  \right) \right|.
\end{equation}
The term $W_N(Z_N)$ has the important property of being superharmonic in each variable:
\begin{equation}\label{eq:superharmonic}
-\nabla^2_{z_i} W_N(Z_N)\geq 0 \mbox{ for all } i. 
\end{equation}
This holds simply because $F$ is analytic and is, in fact, the only property of $W_N$ that is used in our method.

A precursor of the desired  bound \eqref{bound} for $\rho_F$ is the fact that the local density of points in a minimizing configuration for $ H_N(Z_N)$ is everywhere bounded above by $N(\pi\ell)^{-1}$ for large~$N$. This is the core of the proof of the theorem, and a signature of screening properties of the effective plasma. 
To establish it, we introduce and study an auxiliary minimization problem, which is,  mathematically, a cousin of the Thomas-Fermi energy minimization problem for molecules~\cite{Lieb-81b}.

For $K$ fixed points $x_i\in\mathbb R^2$ (\lq\lq nuclei\rq\rq) we define an energy for functions $\sigma$ on $\mathbb R^2$ (\lq\lq electron density\rq\rq) by
\begin{equation} \label{esigma} 
\ETF [\sigma]=-\int_{\R^2} V_{\rm nucl} (x)\sigma(x)\,dx + D(\sigma, \sigma)
\end{equation}
with
\begin{equation} 
V_{\rm nucl}(x)=-\sum_{i=1}^K \log{|x-x_i|}
\end{equation} 
and 
\begin{equation} 
D(\sigma,\sigma')=-\frac{1}{2} \iint_{\R^2 \times \R^2} \sigma(x)\log{|x-x'|}\sigma'(x')\, dx\, dx'.
\end{equation} 
This functional is to be minimized under the subsidiary conditions
\begin{equation} 
\quad 0\leq \sigma \leq 1,\quad \int_{\R ^2}\sigma =K.
\end{equation} 
In physical terms, this model describes a neutral 2D molecule consisting of fixed nuclei and mobile electrons, with Coulomb interactions. The interpretation of the constraint $0\leq \sigma\leq 1$ is that the kinetic energy of the electrons is zero for densities $\leq 1$ and $\infty$ for densities~$>1$.

The basic facts about this TF model are: 
\begin{enumerate}
 \item There exists a unique minimizer, $\TFmin$.
 \item The minimizer has compact support.
 \item Apart from a set of measure zero, $\TFmin$ {\it takes only the values} 0 or 1.
 \item The Thomas-Fermi  equation holds:
\begin{equation} \label{eq:TF eq}
\TFpot(x)=\begin{cases} \geq 0&\hbox{\rm  if } \sigma^{\rm TF}(x)=1\\ 0&\hbox{\rm if } \sigma^{\rm TF}(x)=0 \end{cases}
\end{equation}
where 
$$\TFpot(x)=\Vnuc(x)+\intR\log{|x-x'|}\sigma^{\rm TF}(x')dx'$$
is the total electrostatic potential of the molecule. 
\end{enumerate}

According to the TF equation the support of $\sigma^{\rm TF}$ is the same as the support of the potential
$\TFpot,$ which is continuous away from the ``nuclei''.
Denote by $\STF (x_1,\dots,x_K)$ the open set where $\TFpot$ is strictly larger than $0$. 
Some important properties are:
\begin{enumerate}
 \item The area of $\STF (x_1,\dots,x_K)$ is equal to $K$.
 \item $\STF (x_1,\dots,x_{K-1})\subset \STF (x_1,\dots,x_K)$.
 \item For a single nucleus at $x_1$, $\STF (x_1)$ is the disc with center $x_1$ and radius $\pi^{-1/2}$.
\end{enumerate}

Consider now a scaled version of \eqref{ham},
\begin{equation}\label{eq:scaled ham} 
\cH (X_N)=\frac \pi 2\sum_{i=1}^N|x_i|^2-\sum_{1\leq i<j\leq N} \log{|x_i-x_j|}+\cW(X_N), 
\end{equation}
with $\cW$ symmetric and superharmonic in each variable $x_i\in\mathbb R^2$ and $X_N=(x_1,\dots,x_N)$. A key property of minimizing configurations of $\cH$ is stated in the following:

\begin{lem}[\textbf{Exclusion rule}]\label{lem:excl}\mbox{}\\
Let $X_N^0 = \{x_1^0,\dots,x_N^0\}$ be a minimizing configuration of points for $\cH$. For any subset of points $y_1,\ldots,y_K,y_{K+1} \in X_N ^0$, 
\begin{equation}\label{exrule} 
y_{K+1} \notin \STF (y_1,\dots,y_{K}).
\end{equation}
\end{lem}

The proof of this exclusion rule is sketched in the Appendix. To make the result plausible, observe the following: the first term in~\eqref{eq:scaled ham} is the electrostatic potential generated by a constant background of charge density $-1$. Using Equation~\eqref{eq:TF eq}, the electrostatic potential generated by the points $y_1,\ldots,y_K$ is completely screened by the part of the background potential generated by the region $\STF (y_1,\dots,y_{K})$. What the exclusion rule says is that no other point of a minimizing configuration can lie inside this screening region. The reason is two-fold: 

\smallskip

\noindent $\bullet$ if another point $y_{K+1}$ lay inside the screening region one could decrease the sum of the first two terms in~\eqref{eq:scaled ham} by moving $y_{K+1}$ to any position on the boundary.

\noindent $\bullet$ the last term $\cW$ in~\eqref{eq:scaled ham}, being superharmonic, is generated by a positive charge distribution. One can thus always decrease $\mathcal W$ by moving $y_{K+1}$ to some point on the boundary. Such a move decreases at the same time the sum of the first two terms.

\smallskip

The particular case $K=1$ of Lemma~\ref{lem:excl} goes back to an unpublished theorem of Lieb, used in \cite{RouYng-15}: The minimal distance between points in a minimizing configuration of $\cH$ is not less than $1/\sqrt \pi$. This property  shows that the density of points is in any case bounded above by $4$.  The general exclusion rule for all $K$ implies more. The density is, in fact, asymptotically bounded above by $1$:

\begin{lem}[\textbf{Exclusion rule implies density bound}]\label{lem:excl to dens}\mbox{}\\
For $R>0$ let $n(R)$ denote the maximum number of any points 
$\{y_1,\dots,y_n\}$ that a disk $D(R)$ of radius $R$ can accommodate while respecting the exclusion rule~\eqref{exrule}. Then
\begin{equation}\label{3.14}
\limsup_{R\to \infty}\frac {n(R)}{\pi R^2}\leq 1.
\end{equation}
\end{lem}

We sketch a proof in the Appendix. The main point is that if a large region contains a density of points $>1$ we can reach a contradiction by considering two facts:

\smallskip 

\noindent $\bullet$ the potential $\TFpot$ generated by the points $y_i$ contained in $D(R)$ and the corresponding exclusion set $\STF (y_1,\dots, y_n)$ must vanish at all the points $y_j$ lying outside of $D(R)$, by the exclusion rule and~\eqref{eq:TF eq}. This leads to a uniform upper bound on $\TFpot$ outside of $D(R)$.

\noindent $\bullet$ the same potential is generated by an overall positive charge density, because the total nuclear charge in $D(R)$, which is $ > \pi R^2$ by assumption, is not fully screened by the part of the negative charge density $\TFmin$ lying in $D(R)$, at most equal to the area $\pi R^2$ because $\TFmin\leq 1$. A lower bound on the circular average of $\TFpot$ outside of $D(R)$ follows.

\smallskip

It turns out that, for $R$  large enough, the two estimates obtained in this way, if \eqref{3.14} fails, will contradict one another.  Consequently, \eqref{3.14} must hold.

After scaling, $x\to z=\sqrt{\frac{\pi\ell}N}\,x$,  Lemma 2 applies to the Hamiltonian \eqref{ham} and implies that in any minimizing configuration $\{z_1^0,\dots,z^0_N\}$ of~\eqref{ham} the number of points $z_i^0$ contained in any disc of radius $R \gg N ^{-1/2}$ is not larger than $N(\pi\ell)^{-1}$ times the area of the disc. This is the gist of the proof. From there, the main argument left to conclude the proof of Theorem~\ref{main thm0} is to show that the density bound for ground states of~\eqref{ham} applies also to the Gibbs state~\eqref{eq:plasma analogy}.

In this argument, we use crucially that the temperature $T$ in~\eqref{eq:plasma analogy} scales as $N^{-1}$ so that the Gibbs measure charges mostly ground state configurations for large $N$. Turning this intuition into a proof follows the lines of~\cite[Section~3]{RouYng-15}, see ~\cite[Section~5]{LieRouYng-17} for the  details. In brief, to access the 1-particle density we rely on the fact that a Gibbs state minimizes the free energy of the corresponding Hamiltonian and use a Feynman-Hellmann-type argument: We perturb the Hamiltonian ~\eqref{ham} by adding a term
$\eps\sum_i U(z_i)$ with $U$ of compact support and prove free energy upper and lower bounds for this perturbed Hamiltonian. After dividing by $\eps$  we obtain a bound on $\int U(z)\rho^\eps(z) dz$ where $\rho^\eps$ is the 1-particle density of the Gibbs state for the perturbed Hamiltonian and show that  this tends to $\int U(z) \rho^0(z) dz$ in the limit $\eps\to 0$, where $\rho^0$ is the empirical measure of the ground state.

Strictly speaking, this strategy requires an priori bound ensuring that $\Psi_F$ lives on length scales of order $N^{1/2}$, e.g., 
\begin{equation}\label{eq:mom bound}
\left\langle \Psi_F, L_N \Psi_F \right\rangle \leq C N ^2,  
\end{equation}
where $L_N$ is the total angular momentum operator and $C$ is independent of $N$ and $F$. This asumption is very reasonable physically, and can in fact be eliminated using a suitable localization procedure (\cite{LieRouYng-17}, Section~5.2).

We point out that our density upper bound holds down to the finest possible scale for {\em ground states} of the plasma Hamiltonian~\eqref{ham}, i.e., on length scales $\gg N^{-1/2}$ (see~\cite{RotSer-14} where the corresponding lower bound is proved in the purely Coulombic case $W = 0$). Note that we are here referring to the scaled variables as in ~\eqref{eq:scaledvar}.  When applying this result to {\em Gibbs states} of~\eqref{ham} we have to restrict ourselves to length scales $\gg N ^{-1/4}$ to control the error terms arising from the entropic contribution to the free energy in the Feynman-Hellmann type argument, but this is likely to be due to a technical limitation of our method. It was, indeed, recently proved that, for the purely Coulombic Hamiltonian, the expected microscopic density estimate holds for low temperature Gibbs states~\cite{BauBouNikYau-15,BauBouNikYau-16,Leble-15b,LebSer-16} (see also~\cite{PetRot-16} for ground states of higher dimensional Coulomb and Riesz gases). It remains to be seen whether a combination of our methods with those of~\cite{BauBouNikYau-15,Leble-15b} could improve our results. 
\section{Conclusion}

We have considered perturbations of the Laughlin state that may arise to accommodate external potentials, while keeping the system in the lowest Landau band and preserving the original correlations. We proved rigorously that no such perturbation can raise the particle density anywhere beyond the Laughlin value $1/(\pi \ell)$. This is one of the criteria for the rigidity of the Laughlin state. Our theorem holds on length scales $\gg N^{1/4}$ (in the original, physical variables) which, while large compared to the magnetic length $O(1)$, are microscopic compared to the system's macroscopic size,  $N^{1/2}$.


\bigskip


\noindent\textbf{Acknowledgments:} We received financial support from the French ANR project ANR-13-JS01-0005-01 (N.~Rougerie), the US NSF grant PHY-1265118  (E.~H.~Lieb) and the European Research Council (ERC) under the European Union's Horizon 2020 research and innovation programme (grant agreement CORFRONMAT No 758620). 

\bigskip

\appendix

\section{Proof of the main Lemmas}\label{sec:app}

Here, for completeness, we sketch the proofs of Lemmas~\ref{lem:excl} and~\ref{lem:excl to dens}. The arguments are kept short; full details may be found in the longer paper~\cite{LieRouYng-17}.

\begin{proof}[Proof of Lemma~\ref{lem:excl}]
By symmetry of the Hamiltonian we may, without loss, choose $y_i = x_i ^0, 1\leq i \leq K+1$. Consider then fixing all points but $x_{K+1}^0$. The energy to consider is then
\begin{equation} 
G(x)= \cH(x_1^0,\dots,x_{K}^0,x,x_{K+2}^0,\cdots x_N^0).
\end{equation}
We claim that if $x\in \STF (x_1^0,\dots,x_{K}^0) \equiv \STF$ then there is an $\tilde x\in \partial  \STF$, the boundary of $\STF$, such that $G(\tilde x)<G(x)$. Thus the minimizing point $x_{K+1} ^0$ cannot lie in $\STF$.

To prove the claim, we add and substract a term $-\int_{\STF}\log{|x-x'|}dx'$ to write
$G(x)=\Phi(x)+R(x)$ with
\begin{equation} \label{potential}
\Phi(x)=-\sum_{i=1}^K\log{|x-x_i^0|}+\int_{\STF} \log{|x-x'|}dx'\end{equation}
and
\begin{equation} 
R(x)=\frac{\pi}{2} |x|^2-\int_{\STF} \log{|x-x'|}dx' -\sum_{i=K+2}^N \log{|x-x_i^0|} + W(x) + \hbox{const.}
\end{equation}
Now, $\Phi$ is precisely the TF potential corresponding to \lq nuclear charges' at $x_i^0,\dots x_K^0$. Hence, using~\eqref{eq:TF eq}, $\Phi>0$ on $\STF$ and zero on the boundary $\partial \STF$. 
The first two terms in $R$ are harmonic on $\STF$ when taken together. (The Laplacian applied to the first term gives $2\pi$ and to the second term $-2\pi$ on $\STF$.) The other terms are superharmonic on~$\STF$. Thus, $R$ takes its minimum on the boundary, so there is a $\tilde x\in  \partial \STF$ with $R(x)\geq R(\tilde x)$. On the other hand, $\Phi(x)>0=\Phi(\tilde x)$, so $G(x)>G(\tilde x)$.
\end{proof}

\begin{proof}[Proof of Lemma~\ref{lem:excl to dens}]
 The proof  is by contradiction. Assume that, for some $\delta>0$, there are arbitrary large radii with the property that the disk $D(R)$ contains at least $(1+\delta)\pi R^2$ points. This leads to a contradiction with~\eqref{exrule} and we present here a sketch of the argument. Full details are given in~\cite{LieRouYng-17}.

 By taking the maximal $\delta$ (which is in any case $\leq 3$) we may, without restriction, assume that the density is also at least $(1+\delta)$ in the annulus $\mathcal A$  of width $\delta \cdot R$ around $D(R)$. 
 Since the density is everywhere bounded above by $4$, apart from the points $y_i\in D(R)$, $i=1,\dots,n$ there must be points $y_j\in\mathcal A$, $j=n+1,\dots, m$ in the configuration such that every point in $\mathcal A$ is at most a distance $O(1)$ from one of the $y_j$. The TF potential $\TFpot$ generated by the $y_i$ and the corresponding exclusion set $\STF (y_1,\dots, y_n)$ must vanish at the $y_j$ by the exclusion rule and~\eqref{eq:TF eq}. 
 
On the other hand, after scaling the variables by $R^{-1}$  and extracting a factor $R^2$ one can show that the gradient of the TF potential is uniformly bounded in the scaled annulus $R^{-1}\mathcal{A}$. The distance between the scaled $y_j$ is now $R^{-1}$ so the scaled potential  goes to zero uniformly on $R^{-1}\mathcal{A}$ as  $R\to \infty$. The same holds, then, for the circular average of the scaled potential.

%

We claim, however, that the latter is strictly bounded away from zero close to the radius~$1$ (corresponding to radius $R$ in the unscaled annulus). This follows from Newton's theorem, because the nuclear charge in $D(R)$, which is $(1+\delta)\pi R^2$ by assumption, is not fully screened by the part of the negative charge density $\TFmin$ lying in $D(R)$. This negative charge is at most equal to the area $\pi R^2$ because $\TFmin\leq 1$. There is thus a contradiction for $R$ large enough.
\end{proof}


\end{document}